\documentclass[11pt]{article}

\usepackage{theorem}
\usepackage{fullpage}
\usepackage{amssymb}
\usepackage[geometry]{ifsym}
\usepackage{amsmath}
\usepackage{epsfig}

\newtheorem{theorem}{Theorem}[section]
\newtheorem{corollary}[theorem]{Corollary}
\newtheorem{lemma}[theorem]{Lemma}
\newtheorem{proposition}[theorem]{Proposition}

\newtheorem{observation}[theorem]{Observation}
{\theorembodyfont{\rmfamily} \newtheorem{example}		[theorem]
{Example}}

\def\squarebox#1{\hbox to #1{\hfill\vbox to #1{\vfill}}}
\newcommand{\qed}{\hspace*{\fill}
\vbox{\hrule\hbox{\vrule\squarebox{.667em}\vrule}\hrule}\smallskip}
\newenvironment{proof}{\noindent{\bf Proof:~~}}{\(\qed\)}

{\theorembodyfont{\rmfamily} }

\newcommand{\eat}[1]{}
\newcommand{\sm}{\setminus}
\newcommand{\one}{\downarrow}
\newcommand{\mone}{\uparrow}
\newcommand{\slack}{S}
\newcommand{\na}{Z}
\newcommand{\compose}{\circ}

\newcommand{\rg}[0]{G }
\newenvironment{prevproof}[2]{\noindent {\em {Proof of
{#1}~\ref{#2}:}}}{$\qed$\vskip \belowdisplayskip}

\title{Universally Utility-Maximizing Privacy Mechanisms}
\author{Arpita Ghosh\thanks{Yahoo!~Research, 2821 Mission College Boulevard, Santa Clara, CA.
Email: {\tt arpita@yahoo-inc.com}.}
\and
Tim Roughgarden\thanks{Department of Computer Science, 
Stanford University, 462 Gates Building, 353 Serra Mall, Stanford, CA 94305.
Supported in part by NSF CAREER Award CCF-0448664, 
an ONR Young Investigator Award, an AFOSR MURI grant,
and an Alfred P. Sloan Fellowship.
Email: {\tt tim@cs.stanford.edu}.}
\and
Mukund Sundararajan\thanks{Department of Computer Science,
Stanford University, 470 Gates Building, 353 Serra Mall, Stanford, CA 94305.
Supported by NSF Award CCF-0448664, a Stanford Graduate Fellowship and
an internship at Yahoo! Research.
Email: {\tt mukunds@cs.stanford.edu}. }
}

\begin{document}

\maketitle
\thispagestyle{empty}

\begin{abstract}
A mechanism for releasing information about a statistical database
with sensitive data must resolve a trade-off between utility and
privacy.  Publishing fully accurate information maximizes utility
while minimizing privacy, while publishing random noise accomplishes
the opposite.  Privacy can be rigorously quantified using the
framework of {\em differential privacy}, which requires that a
mechanism's output distribution is nearly the same
whether or not a given database row is included or excluded.  
The goal of this paper is strong and general utility guarantees,
subject to differential privacy.

We pursue mechanisms that guarantee near-optimal utility to every
potential user, independent of its side information (modeled as a
prior distribution over query results) and preferences (modeled via a
loss function).
Our main result is: for each fixed count query and differential
privacy level, there is a {\em geometric mechanism} $M^*$ --- a discrete
variant of the simple and well-studied Laplace mechanism --- that is {\em
simultaneously expected loss-minimizing} for every possible user,
subject to the differential privacy constraint. 
This is an extremely strong
utility guarantee: {\em every} potential user $u$, no matter what its
side information and preferences, derives as much utility from $M^*$
as from interacting with a differentially private mechanism $M_u$ that
is optimally tailored to~$u$.
More precisely, for every user~$u$ there is an optimal mechanism~$M_u$
for it that factors into a user-independent part (the geometric
mechanism~$M^*$) followed by user-specific post-processing that can be
delegated to the user itself.

The first part of our proof of this result characterizes the
optimal differentially private mechanism for a fixed but arbitrary
user in terms of a certain basic feasible solution to a linear program
with constraints that encode differential privacy.  The second part shows 
that all of the relevant vertices of this polytope (ranging over all
possible users) are derivable from the geometric mechanism via
suitable remappings of its range.
\end{abstract}

\section{Introduction}

Organizations including the census bureau, medical
establishments, and Internet companies collect and publish statistical
information~\cite{Census, AOL}. The
census bureau may, for instance, publish the result of a query such as:
``How many individuals have incomes that exceed \$100,000?''. 
An implicit hope in this approach is that aggregate information is
sufficiently anonymous so as not to breach the privacy of any individual. 
Unfortunately, publication schemes initially thought to be ``private''
have succumbed to privacy attacks~\cite{AOL, Netflix,Socialnetwork},
highlighting the urgent need for mechanisms that are \emph{provably}
private. The differential privacy literature~\cite{BigBang, Survey,
MD, Smooth,Learning, Impossibility,  Datamining, SULQ} has proposed a
rigorous and quantifiable definition of privacy, as well as
provably privacy-preserving mechanisms for diverse applications
including statistical queries, machine learning, and
pricing. Informally, for $\alpha \in [0,1]$,
a randomized mechanism is $\alpha$-differentially
private if changing a row of the underlying database---the data of a
single individual---changes the probability of every mechanism output
by at most an $\alpha$ factor. Larger values of $\alpha$ correspond to
greater levels of privacy. Differential privacy is typically achieved
by adding noise that scales with $\alpha$.  While
it is trivially possible to achieve any level of differential privacy,
for instance by always returning random noise, this completely defeats
the original purpose of providing useful information. On the
other hand, returning fully accurate results can lead to privacy
disclosures~\cite{Survey}. \emph{The goal of this paper is to identify,
for each $\alpha \in [0,1]$, the optimal (i.e., utility-maximizing)
$\alpha$-differentially private mechanism.}

\section{Model}

\subsection{Differential Privacy}
\label{sec:dp}

We consider databases with $n$ rows drawn
from a finite domain $D$. Every row corresponds to an individual. Two
databases are \emph{neighbors} if they coincide in $n-1$ rows. A
\emph{count query}~$f$ takes a database~$d \in D^n$ as input and
returns the result $f(d) \in N= \{ 0, \ldots, n\}$ that is the number
of rows that satisfy a fixed, non-trivial predicate on the domain
$D$. Such queries are also called predicate or subset-sum queries;
they have been extensively studied in their own
right~\cite{Impossibility, SULQ, Datamining,Learning}, and form a
basic primitive from which more complex queries can be
constructed~\cite{SULQ}.

A randomized mechanism with a (countable) range $R$ is a function $x$
from $D^n$ to $R$, where $x_{dr}$ is the probability of outputting
the response $r$ when the underlying database is $d$. For $\alpha \in
[0,1]$, a  mechanism $x$ is \emph{$\alpha$-differentially private} if
the ratio $x_{{d_1}r}/x_{{d_2}r}$ lies in the interval
$[\alpha,1/\alpha]$ for every possible output $r \in R$ and pair $d_1,
d_2$ of neighboring databases\footnote{ The usual definition of
differential privacy requires this bound on the ratio for all
subsets of the range. With a countable range, the two definitions
are equivalent.}. (We 
interpret $0/0$ as $1$.)  Intuitively, the probability of every response
of the privacy mechanism --- and hence the probability of a successful
privacy attack following an interaction with the mechanism --- is, up
to a controllable $\alpha$ factor, independent of whether a given user
``opts in'' or ``opts out'' of the database.

A mechanism is \emph{oblivious} if, for all $r \in R$, $x_{d_1r} =
x_{d_2r}$ whenever $f(d_1) = f(d_2)$ --- if the output distribution depends
only on the query result. Most of this paper considers only oblivious
mechanisms; for optimal privacy mechanism design, this is without loss
of generality in a precise sense (see Section~\ref{sec:oblivious}). 
The notation and definitions above simplify for 
oblivious mechanisms and count queries. We can specify an oblivious
mechanism via the probabilities $x_{ir}$ of outputting a response $r
\in R$ for each query result $i \in N$; $\alpha$-differential
privacy is then equivalent to the constraint that the ratios
$x_{ir}/x_{(i+1)r}$ lie in the interval $[\alpha,1/\alpha]$ for every
possible output $r \in R$ and query result $i \in N \sm \{n\}$.

\begin{example}[Geometric Mechanism] 
The \emph{$\alpha$- geometric mechanism} is defined as follows. When
the true query result is $f(d)$, the mechanism outputs $f(d) + Z$. $Z$
is a random variable distributed as a two-sided geometric
distribution: 
$Pr[Z=z]=\frac{1-\alpha}{1+\alpha} \alpha^{|z|}$ for every
integer~$z$. This (oblivious) mechanism is $\alpha$-differentially
private because the probabilities of adjacent points in its range
differ by an $\alpha$ factor and because the true answer to a count
query differs by at most one on neighboring databases.
\end{example}

\subsection{Utility Model} \label{sec:model}

This paper pursues strong and general utility guarantees. Just as
differential privacy guarantees protection against every potential
attacker, independent of its side information, we seek mechanisms that
guarantee near-optimal utility to \emph{every} potential user, independent of
its side information and preferences. 

We now formally define preferences and side information. We model the
preferences of a user via a \emph{loss function} $l$; $l(i,r)$ denotes
the user's loss when the query result is $i$ and the mechanism's
(perturbed) output is $r$. We allow $l$ to be arbitrary, subject only
to being nonnegative, and nondecreasing in $|i-r|$ for each
fixed~$i$. For example, the loss function $|i-r|$ measures mean error,
the implicit measure  of (dis)utility in most previous literature on
differential privacy. Two among the many other natural possibilities are
$(i-r)^2$, which essentially measures variance of the error; and the binary loss
function $l_{bin}(i,r)$, defined as $0$ if $i = r$ and $1$ otherwise.

We model the side information of a user as a
prior probability distribution $\{p_i\}$ over the query results $i \in N$. This prior
represents the beliefs of the user, which might stem from 
other information sources, previous interactions with the mechanism, 
introspection, or common sense. We emphasize that we are {\em not} introducing priors to weaken
the definition of differential privacy; we use the standard definition
of differential privacy (which makes no assumptions about the side
information of an attacker) and use a prior only to discuss the {\em
  utility} of a (differentially private) mechanism to a potential user. 
  
Now consider a user with a prior  $\{p_i\}$ and loss function $l$ and an oblivious mechanism $x$ with range $R$. For a given input $d$ with query result $i =f(d)$, the user's expected loss is $\sum_{r \in R} x_{ir} \cdot l(i,r)$, where the expectation is over the coin flips internal to the mechanism. The user's prior then yields a measure of the mechanism's overall (dis)utility to the the user: 

\begin{equation} \label{eq:utility}
\sum_{i \in N} p_i \sum_{r \in R} x_{ir} \cdot l(i,r).
\end{equation}

This is simply the expected loss over the coin tosses of the mechanism
and the prior.\footnote{The central theorem of choice theory
  (e.g.~~\cite[Chapter 6]{BigFat}) states that every preference relation
over mechanisms that satisfies  reasonable axioms (encoding
``rationality'') can be modeled via \emph{expected} utility, just as
we propose. In particular, this theorem justifies the use of priors
for expressing a rational user's trade-off over possible query
results.} We can then define the \emph{optimal}
$\alpha$-differentially private mechanism for a user as one that
minimizes the user-specific objective function (\ref{eq:utility}).

\subsection{User Post-Processing} \label{sec:remap}

Could a single mechanism be good simultaneously for all users? A
crucial observation for an affirmative answer is that a user has the
power to post-process the output of a privacy mechanism, and that such
post-processing can decrease the user's expected loss.

\begin{example}[Post-Processing Decreases Loss]
Fix a database size $n$ that is odd. Consider a user with the binary
loss function $l_{bin}$ and prior $p_{0} = p_{n}=1/2$, $p_{j} =0$ for
all $j \in N \setminus \{0,n\}$, that interacts with the
$\alpha$-geometric mechanism. Without post-processing, i.e., when the
user accepts the mechanism's outputs at face value, the user's expected
loss is $(2 \cdot \alpha)/(1+\alpha)$. If the user maps outputs of the
geometric mechanism that are at least $(n+1)/2$ to $n$ and all other
outputs to $0$, it effectively induces a new mechanism
with the much smaller expected loss of $\alpha^{(n+1)/2}/(1+\alpha)$.
\end{example}

In general, a (randomized) \emph{remap} of a mechanism $x$ with range $R$ is a
probabilistic function $y$, with $y_{r'r}$ denoting the probability
that the user reinterprets the mechanism's response $r'\in R$ as the
response $r \in R$. A mechanism $x$ and a remap $y$ together induce a
new ($\alpha$-differentially private) mechanism $y \compose x$ with
$(y \compose x)_{ir} = \sum_{r' \in R} x_{ir'} \cdot y_{r'r}$.

We assume that a (rational) user with prior $p$ and loss function
$l$, interacting with a publicly known mechanism $x$, employs a remap
$y$ that induces the mechanism $y \compose x$ that minimizes its
expected loss~\eqref{eq:utility} over all such remaps. 
It is well known (e.g.~\cite[Chapter 9]{BigFat}) and
easy to see that, among all possible (randomized) remappings, the
optimal one follows from applying Bayes rule and then minimizing
expected loss. Precisely, for each response $r$ of $x$, compute the
posterior probability over query results: For every $i \in R$,  $p(i|r) =
(x_{ir} \cdot p_i)  / (\sum_{i' \in R} x_{i'r} \cdot p_{i'})$. Then,
choose the query result $i^* \in R$ that minimizes expected loss
subject to the posterior and set $y_{ri^*} =1$ and $y_{ri} =0$ for $i
\neq i^*$. 
This remap is simple, deterministic, and can be computed efficiently.

\section{Main Result and Discussion}

\begin{figure*}
\hskip 1 in 
\epsfig{angle=-90, file=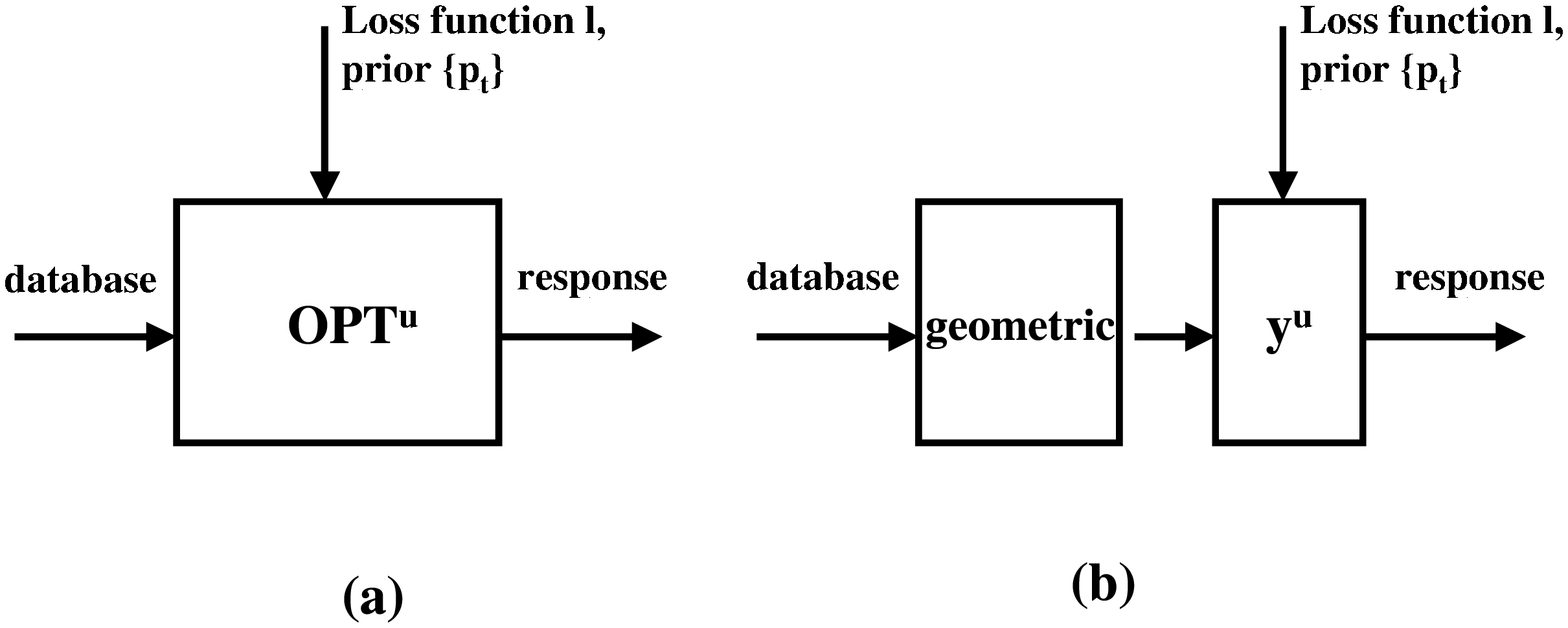, width= 5in}
\label{fig:factor}
\caption{Theorem~\ref{thm:main}. For every rational user $u$, (a) can
be factored into a user-independent part (the $\alpha$-geometric
mechanism) followed by a user-dependent post-processing step (the
optimal remap $y^u$).}
\end{figure*}

Our main result is that for every $\alpha \in [0,1]$, 
{\em the
$\alpha$-geometric mechanism is simultaneously optimal for every
rational user.}

\begin{theorem}[Main Result] \label{thm:main}
Let $x^G$ denote the $\alpha$- geometric mechanism for some database
size $n\geq 1$ and privacy level $\alpha \in [0,1]$, and let $y^u$
denote an optimal remap of $x^G$ for the user $u$ with prior $p$ and
(monotone) loss function $l$. Then $y^u \compose x^G$ minimizes $u$'s expected
loss~\eqref{eq:utility} over all oblivious, $\alpha$-differentially
private mechanisms with range $N$.
\end{theorem}

This is an extremely strong utility-maximization guarantee: {\em
every} potential user $u$, no matter what its side information and
preferences, derives as much utility from the geometric mechanism as
it does from interacting with a differentially private mechanism $M_u$
that is optimally tailored to $u$. We reiterate that the prior from
the utility model plays no role in the definition of privacy, which is
the standard, worst-case (over adversaries with arbitrary
side-information and intent) guarantee provided by differential
privacy. 
We emphasize that while the geometric mechanism is user-independent
(all users see the same distribution over responses), different users
remap its responses in different ways, as informed by their individual
prior distributions and loss functions.  
Rephrasing Theorem~\ref{thm:main}, for every user there is an 
optimal mechanism for it that factors into a user-independent
part---the $\alpha$-geometric mechanism---and a user-specific
computation that can be delegated to the user. (See
Figure~\ref{fig:factor}.)

Theorem~\ref{thm:main} shows how to achieve the same utility as a
user-specific optimal mechanism without directly implementing one.
Direct user-specific optimization would clearly involve several
challenges. 
First, it would require advance knowledge or elicitation of user
preferences, which we expect is impractical in most applications.
And even if a mechanism was privy to the various preferences of its users,
it would effectively need to answer the same query in different ways
for different users, which in turn degrades its differential privacy
guarantee. 

In Theorem~\ref{thm:main}, the restriction to oblivious mechanisms is,
in a precise sense, without loss of generality. (See
Section~\ref{sec:oblivious}.) The restriction to the range $N$
effectively requires that the mechanism output is a legitimate query
result for some database; this type of property is called
``consistency'' in the literature (e.g.~\cite{Consistency}).

\section{Related Work}

Differential privacy is motivated in part by the provable
impossibility of absolute privacy against attackers with arbitrary side
information~\cite{Survey}. 
One interpretation of differential privacy
is: no matter what prior distribution over databases a
potential attacker has, its posterior after interacting with a
differentially private mechanism is almost independent of whether a
given user ``opted in'' or ``opted out'' of the
database~\cite{BigBang, Semantic}.
Below we discuss the papers in the differential privacy literature
closest to the present work; see~\cite{Survey08} for a recent,
thorough survey of the state of the field.

Dinur and Nissim~\cite{Impossibility} showed that for a database
with $n$ rows, answering $O(n \log^2 n)$ randomly chosen subset count
queries with $o(\sqrt{n})$ error allows an adversary to reconstruct
most of the rows of the database (a blatant privacy breach); see
Dwork et al.~\cite{Decoding} for a more robust impossibility result of
the same type.
Most of the differential privacy literature circumvents these
impossibility results by focusing on interactive models where a
mechanism supplies answers to only a sub-linear (in $n$) number of
queries. Count queries (e.g.~\cite{Impossibility,Datamining}) and
more general queries (e.g.~\cite{BigBang,Smooth}) have been studied
from this perspective.

Blum et al.~\cite{Learning} take a different approach by restricting
attention to count queries that lie in a restricted class; they
obtain non-interactive mechanisms that provide simultaneous good
accuracy (in terms of worst-case error) for all count queries from a
class with polynomial VC dimension.
Kasiviswanathan et al.~\cite{Kavi08} give further results for privately
learning hypotheses from a given class.

The use of abstract ``utility functions'' in McSherry and
Talwar~\cite{MD} has a similar flavor to our use of loss functions,
though the motivations and goals of their work and ours are unrelated.
Motivated by pricing problems,
McSherry and Talwar~\cite{MD} design differentially private mechanisms
for queries that can have very different values on neighboring
databases (unlike count queries); they do not consider users with side
information (i.e., priors) and do not formulate a notion of mechanism
optimality (simultaneous or otherwise).

Finally, in recent and independent work, McSherry and Talwar
(personal communication, October 2008) also apply linear programming
theory in the analysis of privacy mechanisms.  Again, their goal is
different: they do not consider a general utility model, but instead
ask how expected error must scale with the number of queries answered
by a differentially private mechanism.

\section{Proof of Main Result } \label{sec:simul}

This section proves Theorem~\ref{thm:main}. The proof has three
high-level steps. 

\begin{enumerate}

\item For a given user $u$, we formulate the problem of determining
the differentially private mechanism that minimizes expected loss
as a solution to a linear program (LP). The objective function of this
LP is user-specific, but the feasible region is not. 

\item We identify several necessary conditions met by every privacy
mechanism that is optimal for some user.

\item For every privacy mechanism that satisfies these conditions, we
  construct a remap $y$ such that $y \compose x^G = x$. By assumption,
  a rational user employs an ``optimal remap'' of $x^G$, so the
  mechanism induced by this map  must be optimal for the user $u$.

\end{enumerate}

Fix a database size $n$ and a privacy level $\alpha$. 
Theorem~\ref{thm:main} is trivially true for the degenerate
cases of $\alpha = 0,1$. So, we assume that $\alpha \in (0,1)$. For
every fixed user with loss function $l$ and prior $p$, the
formulation of privacy constraints in Section~\ref{sec:dp} together
with the objective function~\eqref{eq:utility} yields the following LP
whose solution is an optimal mechanism for this user.

\begin{footnotesize}
\begin{eqnarray}
\nonumber
\mbox{User-specific LP:} \quad \quad && \\ 
\mbox{minimize} && \sum_{i \in N}   p_i \sum_{r \in N} x_{ir} \cdot l(i,r) \\ 
x_{ir} -  \alpha \cdot x_{(i+1)r}  \geq 0 &&  \forall r \in N \setminus \{n\} ,\, \forall i \in N  \label{pos} \\
\alpha \cdot x_{ir} -   x_{(i+1)r}  \leq 0 &&  \forall r \in N \setminus \{n\} ,\, \forall i \in N  \label{neg} \\
\sum_{r \in R} x_{ir} = 1 && \forall i \in N \label{prob} \\
x_{ir} \geq 0 && \forall i \in N, \forall r \in N \label{trivial}
\end{eqnarray}
\end{footnotesize}

Since the LP is bounded and feasible, we have the following
(e.g.~\cite{bert}).

\begin{lemma} \label{lem:lp}
Every user-specific LP has an optimal solution that is a vertex.
\end{lemma}

For the rest of this section, fix a user with prior $\{p_i\}$ and a
loss function $l(i,r)$ that is monotone in $|i-r|$ for every $i \in
N$. Fix a mechanism $x$ that is optimal for this user, and also a
vertex of the polytope of the user-specific LP. Vertices can  be
uniquely identified by the set of all constraints that are tight at
the vertex. This motivates us to characterize the {\em state of
  constraints} (slack or tight) of mechanisms that are optimal for
some user.

We now begin the second step of the proof.  We will view $x$ as a
$(n+1) \times (n+1)$-matrix where rows correspond 
to query results (inputs) and columns correspond to query responses
(outputs). 
We state our necessary conditions in terms of an $n \times (n+1)$
\emph{constraint} matrix~$C$ associated with
the mechanism $x$. Row $i$ of the constraint matrix corresponds to
rows $i$ and $i+1$ of the corresponding mechanism. Every entry of
$C(i,r)$, for $i \in N \sm {n}$, $r \in N$ takes on exactly one of four
values. If $x_{ir} = x_{(i+1)r} =0$ then $C(i,r) = \na$. If $x_{ir},
x_{(i+1)r} \neq \na$, then there are three possibilities. If
$\alpha \cdot  x_{ir} = x_{(i+1)r}$ then $C(i,r) = \one$. If
$x_{ir} = \alpha \cdot x_{(i+1)r}$ then $C(i,r) = \mone$. Otherwise  
$C(i,r) = \slack$.

\begin{figure} 
\begin{scriptsize}
\begin{center}
  \begin{tabular}{|c| c | c | c | c| c|c|}
    \hline
  \mbox{input/output} & 0 & 1 & 2 & 3 & 4 & 5 \\ \hline
  0 &  2/3  &   0 & 1/4  &  1/24  &  1/48 &   1/48 \\ \hline 
  1 &  1/3  & 0 &1/2&   1/12  &     1/24 &     1/24 \\ \hline
  2 &  1/6 & 0 & 1/2  &  1/6 &   1/12 & 1/12 \\ \hline
  3 &  1/12 & 0 & 1/4  & 1/3  &  1/6 & 1/6 \\ \hline
  4 & 1/24  & 0 &  1/8 & 1/6 & 1/3 & 1/3 \\ \hline
  5 & 1/48 & 0 & 1/16  & 1/12 & 1/6 & 2/3  \\ \hline
  \end{tabular}
\end{center}
\caption{The database size $n$ is $5$. 
Figure shows an optimal $1/2$-differentially private mechanism for a
user with prior $1/4, 0, 1/4, 0, 1/4, 1/4$ on the six possible results
and the loss function $l(i,r) = |i-r|^{1.5}$} \label{fig:optA}
\end{scriptsize}
\end{figure}

\begin{figure} 
\begin{center}
\begin{small}
  \begin{tabular}{|c| c | c | c | c| c|c|}
    \hline
  \mbox{input/output} & 0 & 1 & 2 & 3 & 4 & 5 \\ \hline
  0 &  $\one$  &  \na & $\mone$  &  $\mone$  &  $\mone$ &  $\mone$ \\ \hline 
  1 &  $\one$  & \na & \slack&  $\mone$  &     $\mone$ &     $\mone$ \\ \hline
  2 &  $\one$ & \na & $\one$  &  $\mone$ &   $\mone$ & $\mone$ \\ \hline
  3 &  $\one$ & \na & $\one$  &  $\one$  &  $\mone$ & $\mone$ \\ \hline
  4 &  $\one$  & \na & $\one$ & $\one$ &  $\one$ & $\mone$ \\ \hline
  \end{tabular}
\end{small}
\end{center}
\caption{The constraint matrix for the mechanism from Figure~\ref{fig:optA}. } \label{fig:constraintA}
\end{figure}

\begin{example}
Figure~\ref{fig:optA} shows an optimal, $1/2$-differentially private mechanism for a specific
user. Figure~\ref{fig:constraintA} lists the constraint matrix of this
mechanism. This mechanism can be derived from the $1/2$-geometric
mechanism by mapping every negative number to $0$, every number larger
than $5$ to $5$, $1$ to $2$ and the other numbers to themselves.
\end{example}

The constraint matrix is well defined: since
$\alpha<1$, at most one of $\alpha \cdot  x_{ij} = x_{(i+1)j}$ or
$x_{ij} = \alpha \cdot x_{(i+1)j}$ holds (or else both are zero).
Also, using that $\alpha > 0$,
$x_{ij} = 0$ implies that $x_{kj}=0$ for all $k$. 
Thus every column of $C$ is either all $\na$'s, which we
then call a {\em $\na$-column}, or has no $\na$ entries, which we then call
a {\em non-$\na$ column}.

By definition, the constraint matrix encodes the state of all the
privacy constraints~\eqref{pos} and~\eqref{neg}. The constraint matrix
also implicitly encodes the state of the inequality
constraints~\eqref{trivial}.  First, $x_{i,j} > 0$ if and only if the $j$th
column of $C$ is non-$\na$.
Also, we can assume that $x_{ij} < 1$ for all
$i,j \in N$; otherwise, since $\alpha >0$, $x$ has singleton support
and the proof of the theorem is trivial.
Since vertices are uniquely identified by the set of all the
constraints that are tight at the vertex, we have the following.

\begin{observation} \label{fac:char}
Every mechanism that is a vertex of the polytope of the user-specific LP has a unique constraint matrix.
\end{observation}

Let $s$ denote the number of $\slack$ entries of $C$. Let $s_i$ denote
the number of $\slack$ entries in the $i$th non-$\na$ column. Let $S_i
= \sum_{k \leq i} s_i$. Let $z$ denote the number of $\na$-columns
of~$C$. The next few lemmas, Lemma~\ref{lem:all},~\ref{lem:row}
and~\ref{lem:increase}, use the fact that $x$ is an optimal mechanism
to show that the \emph{rows} of $C$ must exhibit certain structure.
Corollary~\ref{cor:manyslacks} and Lemma~\ref{lem:exactly}
additionally use that $x$ is a vertex to show that the total number of
$\slack$ entries of $C$ is equal to the number of
$\na$-columns. Unless otherwise stated, we \emph{ignore $\na$-columns
of~$C$}. The following lemma holds for every feasible mechanism.

\begin{lemma} \label{lem:all}
No row of the constraint matrix is either all $\one$'s or all $\mone$'s.
\end{lemma}

\begin{proof}
Suppose row $i$ of the constraint matrix is all $\one$'s. So, for all
$r \in N$, $\alpha \cdot x_{ir} = x_{(i+1)r}$  and because $\alpha<1$,
$x_{ir} > x_{(i+1)r}$. Since $\sum_{r} x_{ir} = 1$, $\sum_{r}
x_{(i+1)r} < 1$, contradicting the feasibility of the mechanism. 
The other case is analogous.
\end{proof}

The next, key lemma relies on the monotonicity of the loss function
$l$ and the optimality of $x$.

\begin{lemma} \label{lem:row}
Every row of the constraint matrix $C$ has the following pattern: some
$\one$'s followed by at most one $\slack$ followed by some $\mone$'s.
\end{lemma}

\begin{proof}
Suppose not; then there exists a row index $i$ and 
column indices $l < m$ such that $C(i,l) \neq \one$ and
$C(i,m) \neq \mone$. There are two cases: $i \leq (l + m)/2$ and 
$i > (l + m)/2$. 
We prove the first case, and suggest modifications to the proof for
the second. We now define a feasible mechanism $x'$ with strictly
smaller expected loss, contradicting the optimality of $x$. We
multiply the numbers $x_{i'm}$ for all $i' \leq i$ by a $1-\delta$
factor for some small $\delta > 0$ and set $x_{i'l} = x_{i'l}+ \delta
\cdot x_{i'm}$. 

Because $i \leq (l + m)/2$, $|i' -l| \leq |i'-m|$, and, so, for all
$i' \leq i$ the expected loss strictly decreases for strictly
monotone loss functions and priors with support $N$.\footnote{If the
prior does not have full support or if the loss function is not
strictly monotone, we can reach the same conclusion using
perturbations and a limiting argument.}
We now discuss feasibility $x'$. Notice that we modify only two
columns of $x$. The set of constraints~\eqref{prob} are preserved by
feasibility of $x$ and the definition of $x'$. Because $\alpha >0$,
$x_{i'm} >0$ for all $i' \in N$ and for sufficiently small $\delta$
the set of constraints~\eqref{trivial} continue to hold. For all
$j<i$, privacy constraints that involve the numbers $\{x_{jm}\}$
continue to hold as all we are doing is scaling. For all $j<i$,
privacy constraints that only involve the numbers in the sets
$\{x_{jm}\}$ and $\{x_{jl}\}$ continue to hold since differential
privacy is preserved by scaling and adding. Finally, the privacy
constraints involving the numbers $x_{im}, x_{(i+1)m}, x_{il},
x_{(i+1)l}$ are preserved for sufficiently small $\delta$ since
$C(i,l) \neq \one$ and $C(i,m) \neq \mone$.

For the second case, when $i > (l + m)/2$, we interchange the
roles of the columns $l$ and $m$ and modify rows $i' \geq i+1$ rather
than $i' \leq i$, i.e., we multiply $x_{i'l}$ for all $i' \geq i+1$ by
$1-\delta$ for a small, positive $\delta$ and set $x_{i'm} = x_{i'm}+
\delta \cdot x_{i'l}$. The proof can be modified accordingly.
\end{proof}

The next lemma relates adjacent rows of $C$; it uses the previous lemma combined with the fact that each row of $x$ is  a valid probability distribution. 

\begin{lemma} \label{lem:increase}
For all $i \in 0 \ldots n-2$, row $i+1$ of $C$ has at least one more $\one$ than row $i$ of $C$; unless row $i+1$ has a $\slack$ in which case row $i+1$ has at least as many $\one$'s as row $i$.  
\end{lemma}
\begin{proof}
Suppose to the contrary that the pair of rows $i$,$i+1$ violates the condition in the lemma statement, then we will show that row $i+2$ of the mechanism violates the probability constraint~\eqref{prob}, i.e. $\sum_{k}x_{(i+2)k} >1$.

Recall the pattern of a row $i$ of $C$ from
Lemma~\ref{lem:row}. Suppose that row $i$ has $\one$'s in exactly
positions $0 \ldots j-1$. Let $\sum_{k <j} x_{ik} =a$, $x_{ij} =b$,
$\sum_{k >j} x_{ik} =c$ and $x_{(i+1)j} =b'$. Because row $i$ of $x$
must satisfy the probability constraint~\eqref{prob}, we have
\begin{equation}\label{eq:A}
a+b+c=1.
\end{equation}
By assumption, $C(i,k) = \one$, for all $k <j$. By
Lemma~\ref{lem:row}, we have that for all $ k \in j+1 \ldots n$,
$C(i,k) = \mone$ and $C(i,j)$ is either a $\slack$ or a $\mone$. So,
$\sum_{k <j} x_{(i+1)k} = \alpha \cdot a$, and $\sum_{k >j} x_{(i+1)k}
=c/ \alpha$. Because row $i+1$ of $x$ must satisfy the probability
constraint~\eqref{prob}, we have
\begin{equation}\label{eq:B}
\alpha \cdot a+ b' + c/ \alpha=1.
\end{equation}
By assumption the rows $i,i+1$ violate the lemma condition. So, by
Lemma~\ref{lem:row} and the definition of the index $j$, we have
$C(i+1,k) = \mone$ for all columns $k \geq j$. So, $\sum_{k \geq j}
x_{(i+2)k} = 1/\alpha \sum_{k \geq j} x_{(i+1)k} = b'/\alpha +
c/\alpha^2$. Suppose that for each $ k \in 1 \ldots j-1$, $x_{(i+2)k}$
is as small as possible subject to the privacy constraint
(so $x_{(i+2)k} = \alpha \cdot x_{(i+1)k}$); this is the worst case
for our argument. Then, $\sum_{k < j} x_{(i+2)k} = \alpha \sum_{k<j}
x_{(i+1)k} = \alpha^2 \cdot a$. To complete the proof we show that
$\alpha^2 a + b'/\alpha + c/\alpha^2$ is strictly larger than $1$ and
so row $i+2$ of $x$ violates the probability constraint.

Using Equations~\eqref{eq:A} and~\eqref{eq:B} to eliminate $a$ and
$c$, we can show that $\alpha^2 a + b'/\alpha + c/\alpha^2 = 1/\alpha
+ \alpha -1 + b - b'\alpha$, which is at least $1/\alpha + \alpha -1$
because $b - b'\alpha \geq 0$ ($x$ is $\alpha$-differentially
private). Simple algebra shows that $1/\alpha + \alpha -1
>1$ whenever $(\alpha -1)^2 >0$, which holds since $\alpha <1$.
\end{proof}

The next two lemmas relate the number of $\na$-columns of $C$ to the number of $\slack$ entries in it.   

\begin{corollary}
\label{cor:manyslacks}
The number $s$ of $\slack$ entries of $C$ is at least the number $z$ of $\na$-columns of $C$. 
\end{corollary}
\begin{proof}
Suppose there are $a$ $\one$'s and $b$ $\slack$'s in the first row of
the constraint matrix. From Lemma~\ref{lem:all}, the first row cannot
consist entirely of $\mone$'s, i.e. $a+b \geq 1$.  By Lemma
\ref{lem:increase}, the number of $\one$'s in the last row is at least
$a + (n-1) - (s-b)$.  i.e.  at least $n-s$. By Lemma~\ref{lem:all},
the number of $\one$'s in the last row must be at most $n-z$. Chaining
the two inequalities gives us the result.
\end{proof}

Unlike the previous lemmas, the next lemma uses the fact that $x$ is a vertex. 

\begin{lemma} \label{lem:exactly}
The number $s$ of $\slack$ entries of $C$ is equal to the number $z$ of $\na$-columns of $C$. 
\end{lemma}
\begin{proof}
By Corollary~\ref{cor:manyslacks} all we need to show is that $s \leq
z$. Recall that $x$ is the solution to an LP with
$(n+1)^2$-variables. Since $x$ is a vertex of the polytope of the
user-specific LP, it must be at the intersection of at least $(n+1)^2$
linearly independent constraints.

Let us account for them. The $n+1$ constraints~\eqref{prob} are
tight. Exactly $z(n+1)$ constraints of the type specified
by~\eqref{trivial} are tight---these involve variables $x_{ij}$, for
all $i \in N$, such that $j$ is an index of some $\na$-column of
$C$. Thus, of the remaining privacy constraints, at least $(n+1)(n-z)$
must be tight. Because (by definition of $C$) every such tight
constraint corresponds to a unique $\one$ or $\mone$ entry, there are
at least $(n+1)(n-z)$ such entries. Thus, of the $n(n+1-z)$ entries
which are not $\na$'s, at most $n(n+1-z) - (n+1)(n-z) = z$ of the
entries are $\slack$.
\end{proof}

Lemma~\ref{lem:valid} leverages the conditions on the rows of $C$ established by previous lemmas to establish conditions on the columns of $C$ via a counting argument. This completes the second step of the proof of Theorem~\ref{thm:main}. 

\begin{lemma}\label{lem:valid}
For all $i \in 0 \ldots n-z$,  the $i$th non-$\na$ column of $C$ has
  $\mone$'s in rows indexed $0 \ldots i-1+ S_{i-1}
  $, $\slack$'s in the $s_i$ positions that follow and $\one$'s in the
  remaining positions. 

\end{lemma}

\begin{proof} 
We do induction on non-$\na$ columns. The base case: The first column does not contain any $\mone$'s, because then by Lemma~\ref{lem:row}, an entire row contains only $\mone$'s, contradicting
Lemma~\ref{lem:all}. By Lemma~\ref{lem:increase} and
Lemma~\ref{lem:row}, a $\slack$ cannot follow a $\one$ and a $\mone$
cannot follow a $\slack$ in the same column. Therefore, the first
column contains $s_0$ $\slack$'s, followed by $\one$'s.

For the inductive step, assume
column $i-1$ has $\mone$'s and $\slack$'s in rows $0 \ldots i-2 +
S_{i-1} $ and so column $i$ has $\mone$'s in these
positions by Lemma~\ref{lem:row}.
We now show by contradiction that $C(i-1+ S_{i-1},i)= \mone$; assume
otherwise. Let the indicator variable $I$ be~1 if $C( i-1+
S_{i-1},i)= \one$ and~0 if $C( i-1+ S_{i-1},i)= \slack$. By the
induction hypothesis, $C( i-1+ S_{i-1},i-1)= \one$ and so by
Lemma~\ref{lem:row}, row $i-1+ S_{i-1}$ has at least $i +I$ $\one$'s
(we index from $0$). By Lemma~\ref{lem:row}, every one of the rows
from row $i+ S_{i-1}$ to row $n-1$ must either contain an additional
$\one$ or expend one of the $s- S_{i-1} - 1 + I$ remaining $\slack$'s.
Thus row $n$ has at least $i +I$ + $n -i-s +1 -I $ = $n+1-s$
$\one$'s. By Corollary~\ref{lem:exactly}, $s=z$ and so the last row of
$C$ is all $\one$'s, which contradicts Lemma~\ref{lem:all}.

We next show that $C( i+ S_{i-1},i) \neq \mone $ and so with an argument similar to the base case, rows $i+ S_{i-1} \ldots n$  of column $i$ consists of $s_i$ zeros, followed by $\one$'s, concluding the proof of the induction step. Suppose $C( i + S_{i-1},i) = \mone$. By the induction hypothesis, $C(i-1 + S_{i-1},i-1) = \one$, thus row $i-1+ S_{i-1}$ contains at least $i$ $\one$'s. By induction $C( i + S_{i-1},i-1) = \one$. So by Lemma~\ref{lem:row}, row $i + S_{i-1}$ does not contain a $\slack$ \emph{and} contains $i$ $\one$'s; this contradicts Lemma~\ref{lem:increase}. 
\end{proof}

This concludes the second step of the proof. For the third step of the
proof, we work with a  finite-range version of the geometric mechanism
$x^G$, called $G$. $G$ is $w \compose x^G$, where $w$ is a
deterministic remap that maps all negative integers to zero, all
integers larger than $n$ to $n$, and all other integers to
themselves. We show, constructively, that there exists a remap
$\hat{y}^u$ such that $x = \hat{y}^u \compose x^G$. This proves the
theorem because the composed remap $\hat{y}^u \compose w$ applied to
$x^G$ induces $x$: $x = (\hat{y}^u \compose w) \compose
x^G$.\footnote{Incidentally, there do exist vertices of the polytope of the
user-specific LP that are not derivable from $G$ by
remapping. Fortunately, these vertices are not optimal for any user.}
Lemma~\ref{lem:remap} clarifies the structure of columns of mechanisms
induced by deterministic remaps of $G$. Lemma~\ref{lem:geomup}
combines Lemma~\ref{lem:remap} and Lemma~\ref{lem:valid} to prove
Theorem~\ref{thm:main}.

\begin{lemma} \label{lem:remap}
Consider a deterministic map that maps the integer sequence $a \ldots b$ (and no others) to a fixed $l \in N$. Then, the $l$th column of the constraint matrix of the mechanism induced by applying this map to $G$ has $\mone$'s in rows
$0,\ldots,a-1$, $\slack$'s in rows $a,\ldots, b-1$, and $\one$'s in the remaining rows.
\end{lemma}
\begin{proof}
Let the variables $\{g_{ij}\}$ denote the entries of $\rg$ and let
$w_i$ be the entry of in row $i$ of column $l$ of the induced
mechanism. By definition, $w_i =\sum_{j=a}^ b g_{ij}$.

The $j$th column of $\rg$ is a multiple of the
column vector $(\alpha^{j},\alpha^{j-1}, \ldots, \alpha^0, \ldots, \alpha^{n-j-1}, \alpha^{n-j})$. So, for $i,j$, $0 \leq i \leq a-1$, $a \leq j \leq b$, $g_{ij} = \alpha g_{(i+1)j}$, and so, $w_{i} = \alpha w_{i+1}$. Similarly, for $i,j$, $b \leq i \leq n$, $a \leq j \leq b$ $g_{ij} = (1/\alpha) g_{i+1j}$, and so, $w_{i} = (1/\alpha)w_{i+1}$. 

Finally, for every $i$, $a \leq i \leq b-1$, there exist $j,j'$,  $a \leq j,j' \leq b$,  such that       $g_{ij} = \alpha g_{i+1j}$  and $g_{ij'} = (1/\alpha)g_{i+1j'}$. Also $\alpha <1$. Therefore, for $i$ such that $a \leq i \leq b-1$ we have that $\alpha w_i < w_{i+1} < (1/\alpha)w_i$. We have the proof by the definition of constraint matrices.
\end{proof}

The following lemma shows that there exists a remapping of $G$ that
induces a mechanism with constraint matrix
identical to~$C$, the constraint matrix corresponding to
our fixed (vertex) optimal mechanism~$x$.
By Observation~\ref{fac:char}, the induced mechanism must be $x$ and
the theorem is proved.

\begin{lemma} 
\label{lem:geomup}
There exists a remapping of $G$ such that the induced mechanism has a
constraint matrix identical to $C$.
\end{lemma}
\begin{proof}
Define the following (user-specific) deterministic map $\hat{y}$,
which we apply to $G$. Let $k_i$ be the index of the $i$th non-$\na$
column of $C$. Map the integers in $i+ S_{i-1} \ldots i+ S_{i}$ to
$k_i$. We check that $\hat{y}$ is well defined: $S_{i} +1$ responses
of $G$ are mapped to $k_i$, and so, a total of $n+1-z + s$ distinct
responses are mapped. By Lemma~\ref{lem:exactly}, $z=s$, so the map is
well defined for every member of the range $N$ of~$G$.
  
No integer that is an index of a $\na$-column of $C$ is in the range
of the map. So, the constraint matrix of the induced mechanism has the
same set of $\na$-columns as $C$. By Lemma~\ref{lem:valid}, the $i$th
non-$\na$ column has $\mone$'s in rows $0,\ldots, i-1 +S_{i-1}$,
$\slack$'s in rows $i +S_{i-1} ,\ldots, i-1+S_{i}$, and $\one$'s in
the remaining rows. By Lemma~\ref{lem:remap} and the definition of
$\hat{y}$, this is precisely the pattern of the $i$th non-$\na$ column
of the constraint matrix of the induced mechanism.
\end{proof}

\section{Discussion}

\subsection{Uniqueness}

Are there other mechanisms for which an analog of
Theorem~\ref{thm:main} holds? An obvious candidate is the well-studied
Laplace mechanism (release the result with noise added from a
distribution with density function $\epsilon/2 \cdot e^{-\epsilon
  |t|}$, where $\alpha = e^{-\epsilon}$), essentially a continuous
version of the geometric mechanism. Here is an example where the
Laplace mechanism compares unfavorably with the geometric mechanism.
Fix a count query, a database size $n$, and a user with loss function
$l_{err}$ and prior $(1/2,1/2)$ over the two possible results
$\{0,1\}$. The geometric mechanism has (optimal) expected loss
$\alpha/(1+\alpha)$ whereas the Laplace mechanism has expected loss
$\sqrt{\alpha}/2$. The approximation achieved by the Laplace mechanism
tends to $\infty$ as $\alpha$ approaches $0$. Though somewhat
pathological, this example rules out the Laplace mechanism as an answer
to the above question.

So, is the $\alpha$-geometric mechanism the unique mechanism for which
an analog of Theorem~\ref{thm:main} holds? No, because, 
the proofs from the previous section demonstrate that the
range-restricted variant~$G$ also satisfies the guarantee of the theorem.
But a uniqueness result is possible: if we restrict attention to
mechanisms with range~$N$, then~$G$ is the unique simultaneously
optimal mechanism, up to a permutation of the range.

\begin{theorem}\label{thm:unique}
Fix a database size $n\geq 1$ and a privacy level $\alpha \in
[0,1]$. Suppose there is an $\alpha$-differentially private mechanism
$x$ with range $N$ such that there exists remap $y^u$ for every user
$u$ where $y^u \compose x$ minimizes the expected loss of user $u$
(Eq~\ref{eq:utility}) over all oblivious, $\alpha$-differentially
private mechanisms with range $N$. Then, there exists permutation remap
$p$ such that $p \compose x = x^G$.
\end{theorem}

\begin{proof} 
Consider a user with uniform prior over $N$ and the loss function
$l_{bin}$. We can solve the user-specific LP to show that $G$ is the
unique optimal mechanism for the user, where the user is constrained
to accept the mechanism's responses at face value.
Suppose there is a mechanism $G'$ distinct from $G$ that is
optimal for all users, and in particular for this user. As $G$ in the
above sense, $G'$ must induce $G$ on application of the user's optimal
remap. By assumption, the range of $G'$ is $N$, while $G$ is onto $N$.
Recall that optimal remaps are deterministic and so $G'$ must also be
onto $N$ \emph{and} the optimal remap must be a permutation.
\end{proof} 

\subsection{Obliviousness}
\label{sec:oblivious}

Recall that our main result (Theorem~\ref{thm:main}) compares the
utility of the remapped $\alpha$-geometric mechanism only to oblivious
mechanisms.  While natural mechanisms (such as the Laplace
mechanism from~\cite{BigBang}) are usually oblivious, we now justify
this restriction from first principles.  Suppose we deploy
non-oblivious mechanisms. 
Measuring the expected utility of such a mechanism
for a given user requires the user to have a prior over databases, and
not merely over query results.
If a user begins only with a prior over query results (arguably the
most natural assumption if that's all it cares about),
the prior could be extended to one over databases in numerous 
ways.  Singling out any one extension would be arbitrary,
so we consider optimizing worst-case utility over all such
extensions.  Formally, for Lemma~\ref{lem:oblivious} below, we
temporarily replace~\eqref{eq:utility} with the objective function
\begin{equation}\label{eq:obj1}
\max_{p_d \rightarrow p_i} 
\sum_{d \in D^n} p_d \sum_{r \in N} x_{dr} \cdot l(f(d),r).
\end{equation}
In~\eqref{eq:obj1},
$p_d \rightarrow p_i$ indicates that $p_d$ is a prior over $D^n$ that
induces the user's prior $p_i$ over $N$,
meaning $p_i = \sum_{d| f(d) =i} p_d$ for every $i \in N$.  The
following lemma then shows that the restriction to oblivious
mechanisms is without of generality.

\begin{proposition} \label{lem:oblivious}
Fix a database size $n \geq 1$ and privacy level $\alpha$. For every
user with prior $\{p_i\}$ and loss function $l$, there is an
$\alpha$-differentially private mechanism that 
minimizes the
objective function~\eqref{eq:obj1} and is also oblivious.
\end{proposition}

\begin{proof}
Fix a privacy level $\alpha$, a database size $n$ and a user with
prior $\{p_t\}$ and monotone loss function $l$.  Let $x$ be a
mechanism that minimizes~\eqref{eq:obj1} for this user. 
We define a mechanism $x'$ that is oblivious,
$\alpha$-differentially private, and has at most as much
expected loss as $x$ for the objective~\eqref{eq:obj1}. For any
database $d \in D^n$, define $E(d)$ as the set of databases with the
same query result as $d$. For a database $d \in D^n$ and response $r
\in N$, let $x'_{dr}$ be the average of $x_{d'r}$ over the databases
$\{d'|d' \in E(d)\}$; $x'$ is oblivious by definition and also a valid
mechanism, in that it specifies a distribution over responses for
every underlying database.

We now show that $x'$ is $\alpha$-differentially private. Fix two
databases $d_1, d_2 \in D^n$ such that $d_1$ and $d_2$ differ in
exactly one row;  We need to show that $\alpha x'_{d_1r} \leq
x'_{d_2r}$. Assume $f(d_1) \ne f(d_2)$, otherwise the proof is
trivial.

For any database of $E(d_1)$, we can generate all its neighbors
(databases that differ in exactly one row) in $E(d_2)$ by enumerating
all the ways in which we can change the query result by exactly
$1$. For instance when $f(d_1) = f(d_2) +1$, pick one of the
$n-f(d_1)$ rows that satisfy the predicate in $d_1$ and change its
value to one of those that violates the predicate. This process is
identical for all databases of $E(d_1)$, and so for all $d \in
E(d_1)$, the number of neighbors of $d$ that belong to the set
$E(d_2)$ is the same (does not vary with $d$). Similarly, for all $d
\in E(d_2)$, the number of neighbors of $d$ that belong to the set
$E(d_1)$ is the same.

Consider the following set of inequalities that hold because $x$ is
$\alpha$-differentially private: $d \in E(d_1)$, $d' \in E(d_2)$,
where $d_1$ and $d_2$ are neighbors, $ \alpha  x_{dr} \leq
x_{d'r}$. By the argument in the above paragraph, all the databases in
$E(d_1)$ appears equally frequently in the left-hand-side of the above
inequality and all the databases in $E(d_2)$ equally frequently in the
right-hand-side. Summing the inequalities and recalling the definition
of $x'$ completes the proof of privacy.

We now show that $x'$ has better worst case expected loss. Since $x'$
is oblivious, the expected loss is the same for all prior
distributions over databases that induces the prior $\{p_i\}$ over
results. By definition of $x'$, the expected loss is: $\sum_{i\in N}
p_i \, avg_{d| f(d)=i} \sum_{r \in N}  x_{dr} \cdot l(i,r)$. For the
mechanism $x$, we can construct an adversarial distribution over
databases that induces $\{p_i\}$, such that for every $i \in N$, the
weight $p_i$ is assigned entirely to a database $d$ that maximizes
$\sum_{r} x_{dr} \cdot l(f(d),r)$ over ${d | f(d) = i}$. Thus $x$
incurs at least as much expected loss as $x'$.
\end{proof}

An alternative is to model users as having a prior over databases. Though priors on databases induce priors on query results, the converse is not necessarily true. Formally, for Theorem~\ref{thm:imposs} below, we
use the following objective function. 

\begin{equation}\label{eq:obj2}
\sum_{d \in D^n} p_d \sum_{r \in N} x_{dr} \cdot l(f(d),r).
\end{equation}

We show that no analog of Theorem~\ref{thm:main} is possible in this model.\footnote{Proposition~\ref{lem:oblivious} implicitly shows that with the restriction to oblivious mechanisms, Theorem~\ref{thm:main} would continue to hold in this model. However, Theorem~\ref{thm:imposs} implies that obliviousness is no longer without loss of generality.} 

\begin{theorem} \label{thm:imposs}
There exists a database size, level of privacy and two users, each with monotone loss functions and distinct priors over databases such that no mechanism is simultaneously optimal for both of them.
\end{theorem}

We now prove Theorem~\ref{thm:imposs}. Suppose that the domain $D$ is $\{0,1\}$ and we draw databases from $D^3$. Fix a count query that counts the number of rows that are $1$. Fix a  privacy level $\alpha=1/2$. We label databases by the set of rows that have values $1$; in this notation, the result of the query is the cardinality of the set. The first user's prior over databases is $\frac{1}{4}$ on $\{1\}$ and $\{2\}$, $\frac{1}{2} - \epsilon$ on $\{2,3\}$, $\epsilon$ on $\{1,3\}$ and $0$ on all others.  The second user's prior is defined by interchanging the roles of rows $1$ and $2$ in this definition. Both users have the monotone loss function $l(i,r) = |i-r|^{1-\delta}$. There exist small, positive $\epsilon$ and $\delta$, for which the unique (non-oblivious) optimal mechanism (where the user accepts results at face value) for the first user can be shown to be the mechanism $x_1$ defined in
Figure~\ref{fig:M1}. (Only the rows that referred to by our proof are specified.) Similarly, the unique  optimal mechanism for the second user can be derived by interchanging the roles of the rows $1$ and $2$.

\begin{figure} \label{fig:M1}
\begin{center}
  \begin{tabular}{| c | c | c |}
    \hline
    \mbox{Set} & 1 & 2 \\ \hline     
    \{\}   & &    \\ \hline
    \{ 1\} & 11/12 & 1/12   \\ \hline
    \{ 2\}& 2/3 & 1/3   \\ \hline
    \{ 3\} & &    \\ \hline
    \{ 1, 2\} & &  \\ \hline
    \{ 1, 3\}& 5/6 & 1/6 \\ \hline
    \{ 2, 3\}& 1/3 &2/3 \\ \hline
    \{ 1, 2, 3\} & & \\ \hline
  \end{tabular}
\end{center}
\caption{The mechanism $x_1$}
\end{figure}

We will show that there is no $1/2$-differentially private mechanism $x$ that implements $x_1$ and $x_2$
simultaneously. Assume to the contrary: Let $y_1$ and $y_2$ be deterministic 
maps that when applied to $x$ induce $x_1$ and $x_2$ respectively. (We skip the easy modification to the proof that allows $y_1$ and $y_2$ to be randomized.) We first show that $x$, $y_1$ and $y_2$ have the following form:

\begin{lemma}\label{lem:fourcolumns}
Without any loss of generality, $x$ has a range $R =
\{l,m,n,o\}$ of size $4$, $y_1$ maps $l$ and $n$ to $1$ and $m$ and $o$ to $2$ and $y_2$ maps $l$ and $o$ to $1$ and $m$ and $n$ to $2$.
\end{lemma}
\begin{proof}
Both $x_1$ and $x_2$ have range size $2$. Label each member of the range of $x$ in one of four ways based on where it is mapped by $y_1$ ($1$ or $2$) and $y_2$ ($1$ or $2$). All the range points with the same label may be combined into one range point. For instance $o$ consists of all the range points mapped to $2$ by $y_1$ and $1$ by $y_2$. The proof is complete.
\end{proof}

\begin{figure}
\begin{center}
  \begin{tabular}{| c | c | c | c| c|}
    \hline
    \mbox{Set} & l & m & n & o\\ \hline     
    \{\}   & & & &   \\ \hline
    \{ 1\} & $\frac{1}{12}(8-\delta_1)$ & $\frac{1}{12}(1-\delta_1)$ & $\frac{1}{12}(3+\delta_1)$ & $\frac{1}{12}\delta_1$\\ \hline
    \{ 2\} &$\frac{1}{12}(8-\delta_2)$ & $\frac{1}{12}(1-\delta_2)$  & $\frac{1}{12}\delta_2$ &  $\frac{1}{12}(3 + \delta_2)$  \\ \hline
    \{ 3\} & & & &    \\ \hline
    \{ 1, 2\} &  & & & \\ \hline
    \{ 1, 3\}& $\frac{1}{6}(2-\delta_3)$ & $\frac{1}{6}(1 -\delta_3)$ & $\frac{1}{6}(3+ \delta_3)$ & $\frac{1}{6}\delta_3$ \\ \hline
    \{ 2, 3\}& $\frac{1}{6}(2-\delta_4)$ & $\frac{1}{6}(1-\delta_4)$& $\frac{1}{6}\delta_4$ & $\frac{1}{6}(3+\delta_4)$\\ \hline
    \{ 1, 2, 3\} & & & & \\ \hline
  \end{tabular}
\end{center}
\label{f:simmech}
\caption{The mechanism $M$}
\end{figure}

We are now ready to prove Theorem~\ref{thm:imposs}.

\vspace{.1in}

\begin{prevproof}{Theorem}{thm:imposs}
By the previous lemma and the definitions of $x_1$ and $x_2$, the mechanism $x$ must have the form in Figure
\ref{f:simmech}. Because $x$ is $1/2$-differentially private and as $\{1\}$ and $\{2\}$ differ in exactly two rows, columns $n$ and $o$ yield the following inequalities: $3 + \delta_1 \leq 4 \delta_2$, and $3 + \delta_2 \leq 4 \delta_1$. Adding the two inequalities, we have: $6 \leq 3(\delta_1+ \delta_2)$. Because all the entries of $x$ are probabilities, we have $0 \leq  \delta_1, \delta_2 \leq 1$. So it must be that $\delta_1 = \delta_2 = 1$. By a similar argument applied to the databases $\{3$\} and $\{4\}$, we can show that $\delta_3 = \delta_4 = 1$. 

Thus, the probability masses on $l$ when the underlying databases are $\{1\}$ and $\{1,3\}$ are $7/12$ and $1/6$ respectively. But this violates privacy because, the two databases differ in exactly one row but the two probabilities are not within a factor $2$ of each other.
\end{prevproof}

\section{Future Directions}

We proposed a model of user utility, where users are param\-etrized by a
prior (modeling side information) and a loss function (modeling
preferences).  Theorem~\ref{thm:main} shows that for every
fixed count query, database size, and level of privacy, there is a
single simple mechanism that is simultaneously optimal for all
rational users. Are analogous results possible for other definitions
of privacy, such as the additive variant of differential privacy (see~\cite{Survey08})?
Is an analogous result possible for other types of queries or for
multiple queries at once? 
When users have priors over databases (Theorem~\ref{thm:imposs}), 
are any positive results (such as simultaneous {\em approximation})
achievable via a single mechanism?

\section{Acknowledgments}
We thank Preston McAfee, John C. Mitchell, Rajeev Motwani, David Pennock and the anonymous referees.

\end{document}